\documentclass[12pt]{article}
\usepackage[utf8]{inputenc}
\usepackage{amsfonts}
\usepackage{amsmath}
\usepackage{amsthm}
\usepackage{amssymb}
\usepackage{natbib}
\usepackage{hyperref}
\usepackage{tikz}
\usepackage[linesnumbered, ruled, vlined]{algorithm2e}
\usepackage{booktabs}
\usepackage{multirow}
\usepackage[flushleft]{threeparttable}
\usepackage{siunitx}
\usepackage{float}
\usepackage{placeins}
\usepackage{setspace}
\usepackage[singlelinecheck=false, justification=centering]{caption}
\usepackage{color, colortbl}
\usepackage{sectsty}
\sectionfont{\centering}

\newcommand\blfootnote[1]{%
  \begingroup
  \renewcommand\thefootnote{}\footnote{#1}%
  \addtocounter{footnote}{-1}%
  \endgroup
}

\newcommand{\blind}{0}

\makeatletter
\newcommand{\removelatexerror}{\let\@latex@error\@gobble}
\makeatother

\DeclareMathOperator{\Prob}{P}

\newtheorem{theorem}{Theorem}[section]

\newtheorem{lemma}[theorem]{Lemma}
\newtheorem{proposition}[theorem]{Proposition}
\newtheorem{definition}[theorem]{Definition}

\newcommand{\Var}{\mathrm{Var}}
\newcommand{\E}{\mathrm{E}}
\newcommand*{\prob}{\mathrm{P}}
\newcommand{\R}{\mathbb{R}}
\newcommand{\ra}[1]{\renewcommand{\arraystretch}{#1}}

\begin{document}

\def\spacingset#1{\renewcommand{\baselinestretch}%
{#1}\small\normalsize} \spacingset{1}

\title{\bf Unit Testing for MCMC and other Monte Carlo Methods}

\if0\blind
{
  \author{Axel Gandy \and James A. Scott }
  \maketitle
  \blfootnote{Professor Axel Gandy is Chair in Statistics, Department of Mathematics, Imperial College London, SW7 2AZ (e-mail: \href{mailto:a.gandy@imperial.ac.uk}{a.gandy@imperial.ac.uk}). James A. Scott is PhD Candidate, Department of Mathematics, Imperial College London, SW7 2AZ (e-mail: \href{mailto:james.scott15@imperial.ac.uk}{james.scott15@imperial.ac.uk}).}
} \fi

\if1\blind
{
  \bigskip
  \bigskip
  \bigskip
  \maketitle
  \medskip
} \fi

\bigskip 
\begin{abstract}
  We propose approaches for testing implementations of Markov Chain
  Monte Carlo methods as well as of general Monte Carlo methods. Based on statistical hypothesis tests, 
  these approaches can be used in a unit 
  testing framework to, for example, check if individual steps in a Gibbs 
  sampler or a reversible jump MCMC have the desired invariant distribution.  Two exact tests
   for assessing whether a given Markov chain
  has a specified invariant distribution are discussed.  These
  and other tests of Monte Carlo methods can be embedded into a sequential method that
  allows low expected effort if the simulation shows the desired
  behavior and high power if it does not. Moreover, the false
  rejection probability can be kept arbitrarily low. For general Monte
  Carlo methods, this allows testing, for example, if a sampler has a
  specified distribution or if a sampler produces samples with the
  desired mean.  The methods have been implemented in the R-package
  mcunit. 
\end{abstract}

\noindent%
\vfill

\newpage
\spacingset{1.45}

\section{INTRODUCTION}
\label{sec:intro}

Markov chain Monte Carlo (MCMC) methods are the main workhorse of
Bayesian statistics. These methods are used to approximate posterior
expectations which are otherwise analytically intractable. While there
exist numerous diagnostics to assess convergence of the Monte Carlo
estimates to \textit{some} value, few articles address whether
they converge to the \textit{correct} values \citep{geweke_2004, cook_2006, talts_validating_2020}. 

MCMC often requires difficult derivations of
marginal and conditional distributions \citep{geman_geman_1984}, and
derivatives of log densities \citep{roberts_stramer_2002, duane_1987,
  girolami_calderhead_2011}. Increasingly sophisticated algorithms raise the scope for 
analytic errors in these derivations as well as for implementation errors. Testing
for such errors should be an integral and routine part of the workflow
of any Bayesian analysis using MCMC. This article proposes new hypothesis
tests to accomplish this. These tests are unique in being exact; they
have guaranteed false rejection probability which can in principle be made as small
as desired.

MCMC algorithms yield dependent samples, limiting the usefulness of
existing procedures for detecting sampler errors \citep{geweke_2004, cook_2006, talts_validating_2020}. This is because the exact distribution of
the test statistics under the null measure is not known, and as a
consequence, there is no guarantee over the false rejection probability. This has
important practical implications. The obvious consequence is that a
researcher applying the methods cannot always determine whether the test
failed because of dependency between samples, or alternatively because
of actual sampler errors that require further investigation. This could lead to a waste
 of valuable researcher time if they try to find
errors that do not exist. Alternatively, errors could go undetected as they are
explained away by correlation between samples.

One solution to sample dependency is to thin the chain, i.e.\ to
subsample at given intervals to obtain approximately
independent samples. This is the technique suggested in \citet{talts_validating_2020}.
The integrated autocorrelation time is
the number of steps required for a chain to forget
its initial state. If this can be estimated well, then subsampling can be used to yield
effectively independent samples. Unfortunately, reliable estimation of the 
quantity is widely considered challenging \citep{sokal_monte_1997}. The target
distribution is often multi-modal and incomplete sampling can
lead to underestimating the quantity. Even supposing access to a good
estimate, it will often be too large to be of practical
use.

Many sampler errors can be detected in fewer
iterations than required for independent samples. Because independence 
is not required, we can detect these faster than alternate
methods, making the tests more efficient in many practical scenarios.
The two suggested tests use ideas already present in the
literature. One test relies on ideas suggested in \citet{besag_clifford_1989}. 
The theoretical results of this paper are extended by allowing  ties in
the observations and  a more general definition of ranks.  The
other test generalizes a method proposed in \citet{gandy_2016} to test a
specific sampler.

We envisage that the new methods would be particularly useful in unit
testing of MCMC and Monte Carlo methods. Unit testing is a standard
part of the software development process
\citep{runeson_survey_2006}. Individual units of a piece of software are
being tested, in order to demonstrate their functionality.
Frameworks for implementing unit tests are available in many
programming languages \citep{wikipedia_contributors_list_2019}, for example
\cite{wickham_testthat_2011} in the R language.  Tests are generally
re-run after changes to the software, to ensure continued
functionality. There can be a substantial number of tests in any piece
of software - so it is important to keep the computational effort
reasonable for a test that passes. Once a test fails, debugging of the
code will usually be needed to pinpoint (and fix) the source of the
error.

When used for unit testing, the tests for MCMC chains
could be used for individual types of updates of e.g.\ a Gibbs sampler
or of a reversible jump MCMC sampler. The tests are constructed to
test if the chain (or the step of the chain) has the correct invariant
distribution. It is not testing if the chain is recurrent.

When using the above-mentioned methods or other (goodness-of-fit)
tests based on simulated data in unit testing, one faces a trade-off
between the false rejection rate, the power, and the sample size. 
Typically, one would like to have a
(very) low false rejection probability, as investigating potential errors is
time-consuming. Also, as mentioned above, the computational effort if
no errors are present should be low. This immediately places bounds on
the alternatives that one can detect.  We present a sequential method
that improves the position in this regard. It sequentially executes
the test and repeats the test only if the test yields moderate
evidence for departure.
This sequential approach is useful for general Monte Carlo tests and
not just the two MCMC approaches. 

Previous methods introduced to tackle this problem are discussed in
Section \ref{sec:lit}. Section \ref{sec:exact} proposes the new exact
tests for MCMC samplers.  Section \ref{sec:sequential} discusses how
to embed exact tests into a sequential testing procedure to increase
power and reduce the false rejection rate. As mentioned, this is
useful for unit testing and applies to more general Monte Carlo
methods.  Section \ref{sec:sims} presents a simulation study comparing
our approach to previous methods. Conclusions are summarized in Section
\ref{sec:conclusion}.  The tests have been implemented in an easy to
use R-package that immediately slots into the existing unit testing
framework for R \citep{wickham_testthat_2011}. This is available at
\url{https://bitbucket.org/agandy/mcunit}. Proofs can be found in Appendix \ref{sec:proofs}.

\subsection{Related Literature for Testing Samplers}
\label{sec:lit}


\citet{geweke_2004} was the first article to formally consider the
problem of detecting errors in MCMC samplers. Their method
compares samples obtained using two techniques for drawing 
from the joint distribution of parameters and data. The first
simulates directly from the generative model. The second is a Gibbs sampler, alternating
between drawing parameters given data (using the MCMC sampler) and
data given parameters. Z-tests are used to compare estimates of moments
of the joint distribution.  The downside of this approach is that the
Gibbs sampler will generate dependent samples. In practical
applications, the parameters and data can be highly correlated, and a high 
computational effort is required to control the false rejection rate.

\citet{cook_2006} propose tests based on sampled posterior quantiles in the Bayesian framework. 
The authors crucially observe that drawing $\theta$ from the prior and $y$ from the
 likelihood implies that $\theta$ is an exact sample from the posterior given $y$. 
 A sample $\theta_{1:L}$ from this posterior distribution is simulated using the sampler to be tested, 
 and the empirical quantile of $\theta$ is computed among this sample. Unfortunately, the suggested limiting distribution of this quantile is incorrect \citep{gelman_2017}, and the 
proposed tests are not applicable when there is sample dependency, as is the case with MCMC.

\citet{talts_validating_2020} proceed identically to \citet{cook_2006},
but instead of using the empirical quantile of $\theta$ among $\theta_{1:L}$, they compute its 
rank. If the samples are independent and continuous 
then the rank statistic is exactly uniform on $\{1,...,L\}$. Repeating this procedure 
multiple times gives a sample of ranks that can be compared to this uniform distribution. 
Rather than constructing a formal test, the authors advocate
visually assessing goodness of fit using histograms.
The authors propose using thinning to deal with dependent samples when using an iterative simulator like
MCMC. Unfortunately, this leaves the method prone to the
aforementioned problems associated with subsampling Markov chains.

\section{EXACT TESTS FOR ERRORS IN MCMC SAMPLERS}
\label{sec:exact}

In this section, we describe two tests for detecting sampler errors for
MCMC samplers. Analogous tests for simple Monte Carlo methods would be
standard statistical tests such as goodness-of-fit tests.

Assume parameters $\theta \in \Theta$ and data $y \in \mathcal{Y}$ are
modeled as a product of prior and likelihood
$\pi(\theta)p(y\mid\theta)$, and that one can independently draw
parameters from the prior and data from the likelihood.

Further, assume
that the MCMC implementation is designed to work for all possible data
$y\in \mathcal{Y}$.
In a Bayesian analysis
we would observe $y_{obs}$ and construct a Markov chain with kernel
$K_{y_{obs}}$ to estimate expectations of functions with respect to the posterior
$\pi(\cdot \mid y_{obs})$. If the data is implemented as an argument, then the sampler is a
collection of kernels $\{K_y : y \in \mathcal{Y}\}$
such that each $K_y$ is expected to have invariant distribution
$\pi(\cdot \mid y)$.

This motivates the null hypothesis that $K_y$ is
$\pi(\cdot \mid y)$-invariant for all $y \in \mathcal{Y}$.  The tests
do not specifically check $K_{y_{obs}}$, but rather the viability of
the sampler over all possible data values. For example, if only the
kernels corresponding to a null set of data has errors, then the tests
would not be able to detect this.

The null hypothesis will be false if there are errors in the sampler,
broadly characterized as either \textit{design} or
\textit{implementation} errors. Design errors correspond to having a
wrong model for the sampler, and may include mistakes in derived
quantities required for sampling, or a mistake in understanding how
a particular sampler works. Implementation errors refer to an
incorrect execution of a given design, regardless of whether that
design is correct. These are likely to be errors in the
written code.

Both proposed methods are essentially goodness of fit tests which
compare a computed sample of statistics to another distribution.  By
exact, we mean to say that the distribution of the sample is exactly
known under the null hypothesis. We do not mean to say that the
p-value is computed exactly. In practice, cheaper inexact methods may
be used to compute the p-values; for example, using a $\chi^2$ test in
the discrete case. This is of little consequence because the sample
size can be explicitly controlled in the test. For a large enough
sample, the p-value will be as if exact.

Our tests are not designed to investigate the mixing behavior or the
ergodicity of the Markov chain. A Markov chain can be correctly
implemented yet slow mixing. Researchers wishing to diagnose slow
mixing can instead refer to the vast literature on the subject
\citep{cowles_1996}. Properties required for full ergodicity,
including irreducibility and aperiodicity, are typically easy to
establish for continuous distributions, and may require proof
otherwise.

Section \ref{sec:alg1} details a very simple test which uses the
Markov chain to yield samples that should be indistinguishable from
independent samples drawn from the generative model under the
null. This idea generalizes a method described in the supplementary
material of \cite{gandy_2016} to test a specific MCMC
sampler. Section \ref{sec:alg2} considers a more elaborate test based
on uniformity of rank statistics. This uses ideas from
\cite{besag_clifford_1989}.

\subsection{Exact Two-Sample Tests}
\label{sec:alg1}

This method samples from the model in two different ways. The first simply samples 
directly using the generative model, while the second starts by sampling directly, but then propagates the
sample parameters $L$ steps forward using the MCMC sampler. Formally samples are generated with the
sequence of steps
\begin{align*}
\theta' & \sim \pi(\cdot),\\
y'  & \sim p( \cdot \mid  \theta'),\\
\theta & \sim  K^{L}_{y'}(\theta', \cdot).
\end{align*}
$\theta'$ is a perfect sample from $\pi(\cdot \mid y')$, and so initiating the chain at $\theta'$ implies that $\theta$ is also exactly from the posterior under the null hypothesis. Since $y'$ is marginally correct, this implies that $\theta$ is unconditionally a sample from the prior. Moreover, if the procedure is repeated, each sample will be independent. Samples generated this way are described as \textit{fitted} samples, while those generated directly from the model are \textit{direct} samples. Algorithm \ref{alg:prior} details the generation of these samples.

\begin{figure}
\removelatexerror
\begin{algorithm}[H]
\caption{\small General algorithm to perform a two-sample test as described in Section \ref{sec:alg1}.\label{alg:prior}}
\For{$n = 1$ to $N_1$}{
Draw $\tilde{\theta} \sim \pi(\cdot)$\;
Draw $\tilde{y}_n \sim p(\cdot \mid  \tilde{\theta})$\;
Run Markov chain $L$ steps from  $\tilde{\theta}$  to obtain $\tilde{\theta}_{n} \sim K^{L}_{\tilde{y}_n}(\tilde{\theta}, \cdot)$\;
}
\For{$n = 1$ to $N_2$}{
Draw $\theta_{n} \sim \pi(\cdot)$\;
Draw $y_n \sim p(\cdot \mid  \theta_{n})$
}
Compare independent samples $\{(\tilde{\theta}_{n}, \tilde{y}_n)\}$ and $\{(\theta_{n},y_n)\}$\;
\end{algorithm}
\footnotesize
NOTES: $N_1$ and $N_2$ are the number of fitted and direct samples respectively.
\end{figure}

Any appropriate goodness of fit test can be employed to compare the fitted and direct samples. 
Under the null, these are independent and identically from the joint distribution of data and parameters. The most
appropriate test to use will depend on the alternative hypotheses
considered, and so we avoid prescribing a specific test here. 
If the samples are continuous, one may use the two-sample Kolmogorov-Smirnov
test, the Cramer-von Mises test or the Wilcoxon
signed rank test. If discrete, a likelihood ratio test or the Pearson's $\chi^2$-test could be used.
If the form of the prior is particularly simple there may be no need to sample from it, and one could instead use a 
parametric one-sample test.


Algorithm \ref{alg:prior} is similar to that proposed by \citet{geweke_2004}, the key difference being that the data is resampled before each MCMC step. This guarantees independence of samples, which will be useful for controlling the false rejection rate whenever there is high correlation between data and parameters.

The method can be extended by iteratively updating both data and parameters. Line 4 could be replaced by repeating, $L$ times, the step $\tilde{\theta}\sim K_{\tilde{y_n}}(\tilde{\theta}, \cdot)$ followed by $\tilde{y}_n \sim p(\cdot \mid \tilde{\theta})$. This is just a Gibbs sampler and letting $\tilde{\theta}_n$ be the final parameter, $(\tilde{\theta}_n, \tilde{y}_n)$ clearly has the same distribution as  $(\theta_n, y_n)$ under the null. This extension may improve power in certain circumstances, as is shown in Section \ref{sec:second_sim}.

\subsection{An Exact Rank Test}
\label{sec:alg2}

Algorithm \ref{alg:prior} may suffer low power for detecting certain errors. Sometimes, there may 
be mistakes in each conditional which when aggregated are undetectable in the joint distribution 
of data and parameters. An example of this is shown in Section \ref{sec:first_sim}. Here a test is proposed 
that, similar to \citet{talts_validating_2020}, compares a sample of
rank statistics to the uniform distribution. Each statistic is
computed using multiple samples from a single posterior distribution,
and so it may better detect divergences that might be averaged out in the joint.

This comes at the expense of requiring each Markov kernel $K_y$ to be
reversible with respect to $\pi(\cdot \mid y)$. This is not
particularly restrictive: most MCMC algorithms are reversible by
design, because showing reversibility is the easiest way to
prove invariance with respect to a target distribution. Many of the most commonly
used algorithms are reversible, including Metropolis Hastings, Hamiltonian
Monte Carlo and slice sampling. Although samplers using composition of
kernels are not reversible (for example, systematic scan Gibbs
sampling), the constituent kernels can still be tested if they are each
individually reversible.

Rank statistics which break ties in a
vector must be considered, so that the null distribution of
the rank is exactly uniform. The generalization is as
follows.  In the following, $S_n$ is the set of permutations of
$\{1,\dots,n\}$, i.e.\ the set of vectors $s\in \{1,\dots,n\}^n$ such
that $s_i\neq s_j$ for all $i\neq j$.

\begin{definition} \label{def:rank} A function $R: \Theta^n \to S_n$ is an ordinal ranking for vectors $\theta_{1:n} \in \Theta^n$.
\end{definition}
Any function which can assign the same rank to two elements of a vector $\theta_{1:n}$ does not satisfy Definition \ref{def:rank}. 

The general idea behind the test is as follows. First, draw $\theta$ from the prior and $y$ from the likelihood. The kernel $K_y$ is used to draw samples from the posterior, and the rank of $\theta$ among these samples is computed. Replicating this procedure multiple times, the resulting rank statistics will be exactly uniform under the null. Any of a number of goodness of fit tests can then be used. Algorithm \ref{alg:ranks} details the generation of a single rank statistic.

How the posterior samples are drawn has important implications for the uniformity of the rank statistics. Imagine, for example, using the MCMC sampler to realize a Markov chain $\theta_{1:L}$ initiated at $\theta_1 = \theta$. Given some ordinal ranking $R$, the null distribution of $R_1(\theta_{1:L})$ is generally not uniform on $\{1,\ldots,L\}$. Although each element of the chain is of course marginally $\pi(\cdot \mid  y)$, the chain has Markovian dependence and its components are not exchangeable. 

Assuming only reversibility, this can be rectified using a technique suggested in
\citet{besag_clifford_1989}, which is extended here to allow for possible ties in the Markov chain. 
Instead of initiating the chain
from $\theta$, sample $M$ uniformly in $\{1,...,L\}$ and let $\theta_M = \theta$.
Then run the chain twice, once forward $L-M$ steps from $\theta_M$, and then backwards $M-1$
steps from $\theta_M$. By Proposition \ref{prop:uniform_null},
$R_M$ will be exactly uniform under the null. Before giving this proposition, a generalization of the
Lemma of \cite{besag_clifford_1989} is stated.

\begin{lemma} \label{lemma:besag} Suppose $R(\theta_{1:L})$ is a random vector with
  values in $S_L$. If $M \sim \text{Uniform}\{1,...,L\}$ independently
  of $R(\theta_{1:L})$ then $R_M(\theta_{1:L})$ is uniformly distributed on $\{1,...,L\}$.
\end{lemma}

\begin{proposition} \label{prop:uniform_null} Let $R_M(\theta_{1:L})$ be the rank
  statistic returned from Algorithm \ref{alg:ranks}. If for every
  $y$ the kernel $K_y$ is $\pi(\cdot\mid y)$-reversible then $R_M(\theta_{1:L}) \sim \text{Uniform}\{1,...,L\}$.
\end{proposition}

\begin{figure}
\removelatexerror
\begin{algorithm}[H]
Draw $M \sim \text{Uniform}\{1, ..., L\}$\;
Draw $\theta_{M} \sim \pi(\cdot)$\;
Draw $y \sim p( \cdot \mid  \theta_{M})$\;
Choose an ordinal ranking $R$ such that $R$ and $M$ are independent\;
\For{$l = 1$ to $M-1$}{
$\theta_{M-l} \sim K_y(\theta_{M-l+1}, \cdot)$;
}
\For{$l = M+1$ to $L$}{
$\theta_{l} \sim K_y(\theta_{l-1}, \cdot)$;
}
\Return $R_M(\theta_{1:L})$
\caption{\small Computing a rank statistic using the method described in Section \ref{sec:alg2}.}
\label{alg:ranks}
\end{algorithm}
\footnotesize
NOTES: $L$ is the number of MCMC samples to use.
\end{figure}

The canonical example of an ordinal ranking that we have in mind first
maps each component of $\theta_{1:n}$ to the real line with a function $h:\Theta \to \R$,  computes the ranks of
$h(\theta_1), \ldots, h(\theta_n)$, breaking ties in some order. 

Importantly, the ordinal ranking in Algorithm \ref{alg:ranks} can be chosen based on any quantity which is independent of $M$. This allows, for example, randomly breaking ties as follows. If you had a collection of `canonical' rankings, with the only difference between them being that the order of breaking ties is different, you could uniformly select a ranking from this set, thus breaking ties randomly. The ranking could also be selected based on $y$ because it is independent of $M$. Specifically, the function $h$ could be the likelihood function mapping $\theta \mapsto p(y \mid \theta)$. This particular statistic is used in the simulation study of Section \ref{sec:first_sim}.

The proposed test may be generalized. In lines 5 and 7 of Algorithm \ref{alg:ranks} one could, with some fixed probability, update $y$ given the current value of $\theta$ rather than updating $\theta$ using the Markov kernel. This would give samples on the joint space which can be compared using an ordinal ranking $R: (\Theta \times \mathcal{Y})^L \mapsto \mathcal{S}_L$. Proposition \ref{prop:uniform_null} still holds because this is simply testing a random scan Gibbs sampler on the joint space of parameters and data, which is of course reversible under the assumptions of the proposition. This can improve power in detecting certain subtle errors, as is shown in Section \ref{sec:second_sim}.

Another extension is to replace $K_y$ by $K_y^k$ for some $k > 1$. This has the effect of thinning the chain and reducing autocorrelation, and will be useful to increase power against more subtle alternatives. The important point, however, is that such thinning is not required for the null distribution of the ranks to be correct. Extending Algorithm \ref{alg:ranks} to multiple testing is simple.

\section{SEQUENTIAL IMPLEMENTATION FOR UNIT TESTS}
\label{sec:sequential}

Unit tests should have a low false rejection probability and a reasonable computational effort if the sampler works. Moreover, the tests ought to have high power and if errors exists, we should be willing to spend a larger effort detecting them. Here, a sequential testing procedure is proposed which achieves these goals. 

Algorithm \ref{alg:sequential} immediately rejects the null under very low p-values, does not reject
the null for higher p-values and continues simulations for p-values
that are low, but not extremely low. The method runs for a maximum of $k$ steps, and multiplies the sample size by $\Delta$ after the first iteration, which serves to detect errors more easily in subsequent iterations.

There are a large number of possible variations on Algorithm \ref{alg:sequential}. For example, one could define the probability of early rejection via ``spending sequences'' as in \cite{gandy_sequential_2009}. If using Algorithms \ref{alg:prior} or \ref{alg:ranks} to generate the p-values, instead of adjusting the number of chains (through $\Delta$), one could instead increase the amount of thinning within chains. This would also raise the power in subsequent iterations. 

\begin{figure}
\removelatexerror
\begin{algorithm}[H]
  $\beta_1$=$\alpha/k$\;
  $\gamma=\beta_1^{1/k}$\;
  \For{$i = 1$ to $k$}{
    $p^{(i)}$=vector of p-values from one of the algorithms (sample size $n$)\;
    $q_i=\min p^{(i)} * d$\;
    \lIf{$q_i\leq\beta_i$}{\KwRet{fail}}
    \lIf{$q_i>\gamma+\beta_i$}{break}
    $\beta_{i+1}=\beta_i/\gamma$\;
    \lIf{$i=1$}{$n=\Delta n$}
  }
  \KwRet{OK}\;
\caption{\small Sequential wrapper around the methods.}
 \label{alg:sequential}
\end{algorithm}
\footnotesize
NOTES: $d$, is the dimension of the p-value vectors $p^{(i)}$; $\alpha$, is the overall desired false rejection rate; $k$, the maximum number of sequential steps; $\Delta$, the factor by which to multiple the sample size after the first iteration.
\end{figure}

As mentioned, the proposed method has an overall false rejection rate of at most $\alpha$, as the following theorem shows. 
\begin{theorem} \label{thm:size_bound}
    Suppose $p^{(1)},\dots,p^{(k)}$ are independent d-variate random
    vectors with values in $[0,1]^{d}$.  If
    $\Prob \{ p^{(i)}_j\leq p \}\leq p$ for all $p\in [0,1]$ and
    $i= 1,\dots,k$, $j= 1,\dots, d$ then
    $\Prob \{\text{fail} \}\leq \alpha$.
\end{theorem}
The added effort of Algorithm \ref{alg:sequential} compared to the non-sequential case is modest if the p-values are generated under the null. Assuming that they are exactly uniform, the expected increase in effort under the null for general $k$ compared to $k=1$ is

\begin{equation}
\label{eq:expected_effort}
\Delta \sum_{i=1}^k\gamma^{i-1}=\Delta \gamma \left( \frac{1-\gamma^{k-1}}{1-\gamma} \right).
 \end{equation}
More generally, if only the inequality for p-values under the null is assumed (i.e.\ the probability of a p-value being below any bound q is at most q), then the expected increase in effort is bounded from above by
    \begin{equation}
      \label{eq:effort_bound}
\Delta \sum_{i=2}^k\prod_{j=1}^{i-1}(\gamma+\beta_{j}).
    \end{equation}
Motivated by the simulation study in Appendix \ref{sec:tuning} the default values for Algorithm \ref{alg:sequential} 
in the R-package \emph{mcunit} are $\alpha=10^{-5}$, $k=7$, and $\Delta=4$.
This leads to $\gamma\approx 0.15$, and $\beta_1\approx 1.4\cdot 10^{-6}$,
$\beta_2\approx 9.8\cdot 10^{-6}$, $\beta_3\approx 6.6\cdot 10^{-5}$,
$\beta_4\approx 4.6\cdot 10^{-4}$, $\beta_5\approx 3.1\cdot 10^{-3}$
$\beta_6\approx 2.1\cdot 10^{-2}$, $\beta_7=\gamma\approx 0.15$.

For these default parameter values, both \eqref{eq:expected_effort} and \eqref{eq:effort_bound} give the expected additional effort at around 68.5\%. For $\Delta=1$, and the other parameters being unchanged, both formulas give around 17.1\%. The difference between the two formulas is negligible when $\alpha$ is chosen to be small. 
%

\section{SIMULATIONS}
\label{sec:sims}

To demonstrate the performance of the proposed and existing tests, this section presents the results of a power analysis using a stylized model, and a sampler in which errors have been intentionally introduced. The tests considered are exact two-sample and rank tests, and the methods of \citet{geweke_2004} and \citet{talts_validating_2020}. Although \citet{talts_validating_2020} propose graphically checking the distribution of their rank statistics to the uniform distribution, here a formal Kolmogorov-Smirnov test is used to allow consistent comparisons with other methods.

Consider the model
\begin{equation}
\label{eq:stylised_model}
y \sim \theta_1 + \theta_2 + \epsilon,
\end{equation}
where $\theta := (\theta_1, \theta_2)$ is apriori independent, zero-mean normal with standard deviation $\sigma = 10$. The white noise term $\epsilon$ is independent from $\theta$ and also zero-mean normal but with variance $\sigma_{\epsilon}^2 = 0.1$ . While inference is easy in this model, we consider drawing posterior samples of $\theta$ using a Gibbs sampler. The posterior conditional distributions for $\theta_1$ and $\theta_2$ are normal with expectations
 \begin{equation} \label{eq:conditional_expectation}
\E[\theta_i \mid y, \theta_{j}] = \frac{\sigma^2}{\sigma^2_{\epsilon} + \sigma^2}\left(y - \theta_j \right),
\end{equation}
and variances
\begin{equation} \label{eq:conditional_variance}
\Var(\theta_i \mid y, \theta_j) = \frac{1}{\frac{1}{\sigma_{\epsilon}^2} + \frac{1}{\sigma^2}}.
\end{equation}
The small $\sigma_{\epsilon}$ induces high correlation between $\theta_1$ and $\theta_2$ in the posterior distributions, and so the Gibbs sampler will mix slowly.

\subsection{Mistakes in Full Conditionals}
\label{sec:first_sim}

Two correctly implemented samplers are considered; one uses random scan of the two coordinates, with the other using systematic scan. Three erroneous samplers, all of which use random scan, are also considered. The first two have mistakes in the conditional expectations and variances respectively; $y - \theta_j$ is replaced with $y + \theta_j$ in \eqref{eq:conditional_expectation}, and in \eqref{eq:conditional_variance} the variance terms are replaced with the corresponding standard deviations. The final mistake considered truncates each conditional distribution either to the left or right of its posterior mean. The decision to truncate left or right is random for each distribution. 

Table \ref{tab:first_sim} presents the results of the power analysis. Each entry records an empirical rejection rate for a given test function(s) and scenario, computed by repeating the test $10^4$ times. The nominal false rejection rate of each test was set to $\alpha = 0.01$. Sequential versions of Algorithms \ref{alg:prior} and \ref{alg:ranks} were used because they were found to have higher power than the non-sequential versions. All methods were calibrated to have comparable computational budgets. Please refer to the table for details of all simulation parameters.

 \definecolor{graylight}{rgb}{1.0,1.0,1.0}
 \definecolor{grayllight}{rgb}{1.0,1.0,1.0}
 
 \begin{table*}
 \ra{1.2}
\centering
\scalebox{0.8}{
 \begin{threeparttable}
 \caption{\small Empirical rejection rates from the power analysis described in Section \ref{sec:first_sim}. \label{tab:first_sim}}
\begin{tabular}{lrrrrrr}
\toprule
\toprule
\multirow[t]{2}{*}{Test Function}                                         			& \multicolumn{2}{l}{Correct} 		& 	&	 \multicolumn{3}{l}{Errors}               \\ 
                                                                                                                             \cmidrule{2-3}  \cmidrule{5-7} 
		           										& \multicolumn{1}{l}{Rand. Scan}      		&  \multicolumn{1}{l}{Sys. Scan} 	&        	& \multicolumn{1}{l}{$\E$} & \multicolumn{1}{l}{$\Var$} & \multicolumn{1}{l}{Truncated} \\ \midrule
\rowcolor{graylight}
\multicolumn{7}{l}{\textit{Sequential exact two-sample test with $\Delta = 2$ and $k = 3$.}}	\\
\rowcolor{graylight}
$\theta_1$                										& 	0.009 		&  	 	0.010 		&	& 	1.000 		& 	0.007		&	0.008		\\
\rowcolor{graylight}
$\theta_1^2$              										& 	0.008 		&       	0.009        	&	& 	1.000       		&	0.009      		&	0.011			\\
\rowcolor{graylight}
$\theta_1 \theta_2$       									& 	0.008  		&     	 	0.008     		&	& 	1.000      		& 	0.010       		&	0.011			\\
\rowcolor{graylight}
$\pi(\theta)$             										& 	0.010   		&       	0.011        		&	& 	1.000       		& 	0.008        	&	0.009		\\
\rowcolor{graylight}
$p(y \mid \theta)$ 										& 	0.010  		&       	0.009         	&	& 	1.000       		& 	1.000      		&	0.007		\\
\rowcolor{graylight}
\textbf{All} \tnote{a}         									& 	\textbf{0.007} 	& 		\textbf{0.009} 	&	&	\textbf{1.000}	& 	\textbf{1.000} 	&	\textbf{0.006}	\\
\rowcolor{grayllight}
\multicolumn{7}{l}{\textit{Sequential exact rank test with $\Delta = 2$ and $k = 3$.}}	\\
\rowcolor{grayllight}
$\theta_1$				                						& 	0.009       		&    	 	0.885         	& 	& 	1.000             	& 	0.149          	&	0.869		\\
\rowcolor{grayllight}
$\theta_1^2$              										& 	0.009 		&     		0.869          	&	& 	1.000             	& 	0.163          	&	0.868		\\
\rowcolor{grayllight}
$\theta_1 \theta_2$       									& 	0.008   		&     		0.155          	&	& 	1.000             	& 	0.731          	&	1.000		\\
\rowcolor{grayllight}
$\pi(\theta)$             										& 	0.009      		&     	 	0.158          	&	& 	1.000             	& 	0.738          	&	1.000		\\
\rowcolor{grayllight}
$p(y \mid \theta)$										&	0.012		&		0.012		&	&	1.000		&	1.000		&	0.010		\\
\rowcolor{grayllight}
\textbf{All}  \tnote{a}               								& 	\textbf{0.008} 	&  		\textbf{0.769} 	&	& 	\textbf{1.000}  	& 	\textbf{1.000}	& 	\textbf{1.000}	\\
\rowcolor{graylight}
\multicolumn{7}{l}{\textit{\citet{geweke_2004}.}}	\\
\rowcolor{graylight}
$\theta_1$                										& 	0.310      	 	&    		0.235          	&	& 	1.000             	& 	0.278          	&	0.303		\\
\rowcolor{graylight}
$\theta_1^2$              										& 	0.322   		&     		0.276          	&	& 	1.000             	& 	0.119          	&	0.329		\\
\rowcolor{graylight}
$\theta_1 \theta_2$       									& 	0.105     		&      		0.071          	&	& 	1.000             	& 	0.101          	&	0.106		\\
\rowcolor{graylight}
$\pi(\theta)$             										& 	0.226 		&    		0.206          	&	& 	1.000             	& 	0.284          	&	0.220		\\
\rowcolor{graylight}
$p(y \mid \theta)$ 										& 	0.010     		&     		0.013          	&	& 	1.000             	& 	1.000          	&	0.010		\\
\rowcolor{graylight}
\textbf{All}   \tnote{a}             								& 	\textbf{0.523} 	& 		\textbf{0.441}	& 	& 	\textbf{1.000}	& 	\textbf{1.000} 	& 	\textbf{0.516}	\\ 
\bottomrule
\end{tabular}
\begin{tablenotes}
\footnotesize
\item NOTES: The exact two-sample tests ran with $L = 5$ and $N_1 = N_2 = 5 \times 10^2$, and KS tests were used to compare the two samples of the test statistic(s). The exact rank tests ran with $L = 5$ and had $5 \times 10^2$ simulated rank statistics, using a $\mathcal{X}^2$-test to test the ranks for uniformity. \citet{geweke_2004} used thinning of $5$ and $6 \times 10^2$ MCMC samples.
\item[a] Refers to using all aforementioned test functions and a Bonferroni correction for multiple testing.
\end{tablenotes}
 \end{threeparttable}
}
\end{table*}

As expected, the exact two-sample test (Algorithm \ref{alg:prior}) achieves the nominal rate of $0.01$ for both random scan and systematic scan. The test always detected the wrong expectations and variances. However, only the data likelihood proved able to detect the wrong variance. The variance error only changes the marginal of $\theta$ slightly, and so any test using a statistic involving only the parameters will require many samples to detect the error. This illustrates the importance of considering statistics on both data and parameters, rather than just parameters when using this two-sample test. Notice that the truncation was undetected. Even though all conditional distributions are wrong, the joint distribution of parameter and data is indistinguishable from the correct joint. Therefore, the test could never have higher power than the nominal rejection rate for the truncation error.

The exact rank test, as described in Section \ref{sec:alg2}, achieved the nominal rate for the random scan Gibbs sampler, however was unable to do so for systematic scan. This is expected as the systematic scan sampler is not reversible. The multiple test always detected the wrong expectation, variance and truncation.  

The correlation between the data and parameters poses a problem for the method of \citet{geweke_2004}, and the false rejection rate is too high. The test rejects the correct samplers roughly half of the time. Again, the truncation cannot be detected by this method, for the same reason as for the exact two-sample test.

Finally, Algorithm 2 from \citet{talts_validating_2020} was run using $10^3$ initial steps in each chain to estimate the effective sample size. Because the posterior correlation is high in this model, the effective sample size was overestimated and the false rejection rate was entirely uncontrolled. Given that the errors can be detected very easily, this method is highly inefficient in the cases considered.
 
 \subsection{Mistakes in Assumed Prior}
 \label{sec:second_sim}
 
The second simulation investigates the power of the tests when mistakes are made in the assumed prior for $\theta$. In all cases considered, the prior is bivariate normal with common mean $\mu$, standard deviation $\sigma$ and correlation $\rho$. As described at the beginning of Section \ref{sec:sims} the correct version corresponds to $\mu = 0$, $\sigma = 10$ and $\rho = 0$. Three erroneous priors are considered; a mean shift to $\mu = 10$, a variance scale to $\sigma = 5$, and dependency with $\rho = 0.5$. As before, all tests were parameterized to have comparable computational effort and the nominal false rejection rate set to $\alpha = 0.01$. The results are displayed in Table \ref{tab:second_sim}, which also details the simulation parameters. 

 \begin{table*}
 \ra{1.2}
\centering
\scalebox{0.8}{
\begin{threeparttable}
\caption{\small Empirical rejection rates for the power analysis described in Section \ref{sec:second_sim}. \label{tab:second_sim}}
\begin{tabular}{lrrrr}
\toprule
\toprule
Test                             	& 	Correct 	& 	$\mu = 10$ 	& 	$\sigma = 5$ 	& 	$\rho = 0.5$ 	\\ 
\midrule
\rowcolor{grayllight}
Seq. Exact Two-Sample 	& 	0.007   	& 	0.018      		& 	0.826        	& 	0.049        	\\
\rowcolor{grayllight}
Seq. Exact Rank 		& 	0.011   	& 	0.012     		& 	1.000        		& 	0.551        		\\
\rowcolor{grayllight}
Exact Two-Sample         	& 	0.008   	& 	0.012      		& 	0.601        		& 	0.025        	\\
\rowcolor{grayllight} 
Exact Rank         		& 	0.011   	& 	0.009      		& 	1.000        		& 	0.316       		\\
\rowcolor{grayllight}
\citet{geweke_2004}  	& 	0.101   	& 	0.909      		& 	1.000        		& 	0.229        	\\
\rowcolor{grayllight}
\citet{talts_validating_2020}       		& 	1.000   	& 	1.000      		& 	1.000        		& 	1.000        		\\ 
\bottomrule
\end{tabular}
\begin{tablenotes}
\footnotesize
\item NOTES: Reported results are for multiple testing using all test functions shown in Table \ref{tab:first_sim}. The seq. two-sample test used $L=50$ and $N_1=N_2= 10^3$, while the seq. rank test used $L= 10$, $5$ thinning steps between samples and $10^3$ rank statistics. Both versions used $\Delta = 2$ and $k = 3$. The non-sequential versions were adjusted to achieve a similar computational time under the null. \citet{geweke_2004} used $10^3$ samples with thinning of $50$, and \citet{talts_validating_2020} used $10^2$ initial steps to estimate ESS.
\end{tablenotes}
\end{threeparttable}
}
\end{table*}

Both the exact two-sample and rank tests did well to maintain the nominal rate, and had high power in detecting the scaled variance. They were unable to detect the mean shift because the prior is uninformative and has little effect on the posterior distributions. It seems that the joint distribution tests of \citet{geweke_2004} has a power advantage here because the marginal distribution of the parameters in the samples will tend to the specified wrong prior as the number of MCMC steps goes to infinity. Nonetheless, the false rejection rate is far above the nominal level in our simulation. The method was also worse than Algorithm \ref{alg:ranks} at detecting the dependency.  It appears that the joint distribution tests of \citet{geweke_2004} can perform comparatively well when errors in individual posterior distributions are subtle, but aggregate in such a way that they are detectable in the simulated joint distribution. The errors in this section are designed such that the \citet{geweke_2004} method will, if run long enough, recover a specified (wrong) joint distribution. In more general cases, it is not clear that the errors will be so easily detectable in the joint.
\citet{talts_validating_2020} performed poorly, again due to the autocorrelation in the Gibbs sampler. It rejected every scenario in every test. Obtaining reasonable results using this method would require much more computation than required the other tests.

Finally, we demonstrate how to improve the power of the exact tests for the above analysis. Recall the extension to the two-sample test where the data is resampled each time $\theta$ is updated.  The rank test described in Section \ref{sec:alg2} was also extended, so that with probability $0.5$, the data $y$ rather than $\theta$ is updated in line 5 and 7 of Algorithm \ref{alg:ranks}. The power of these generalized methods are estimated under the wrong prior expectation $\mu = 10$ introduced above. The two-sample test was parameterized to used $L=2 \times 10^3$ and $N_1=N_2=10^3$, while the rank test used $L=10$, thinning of $200$ and $10^3$ rank statistics. The empirical ejection rates were $99.2\%$ and $97.2\%$ respectively, computed by replicating each test $10^3$ times. The empirical rejection rates for the original tests were $30.8\%$ and $29.7\%$ respectively, when both were parameterized to have similar computation time. The power improves because the generalizations define Gibbs samplers on the joint space, and so they have higher power in detecting the `aggregated' error in the joint. Nonetheless, for this amount of computation the method of \citet{geweke_2004} achieved the nominal false rejection rate, and had power of $100\%$.

\section{APPLICATION: REVERSIBLE JUMP-MCMC FOR SIGNAL DECOMPOSITION}
\label{sec:rjmcmc}

\citet{andrieu_joint_1999} propose a Bayesian method to jointly detect and estimate sinusoids in a noisy signal. The number of sinusoids making up the signal is unknown, and the authors propose a reversible-jump MCMC (RJ-MCMC) algorithm to explore the space of models consisting of different numbers of sinusoids. This seminal work precipitated a number of studies applying RJ-MCMC to signal processing problems. Many of these relied on the same Metropolis-Hastings-Green acceptance ratio for `birth' and `death' moves that was derived in \citet{andrieu_joint_1999}, including \citet{andrieu_robust_2001, andrieu_bayesian_2001, larocque_reversible_2002, larocque_particle_2002, ng_wideband_2005, davy_bayesian_2006, hong_joint_2010, schmidt_infinite_2010, rubtsov_time-domain_2007}.

\citet{roodaki_comments_2013} demonstrate that this ratio is in fact erroneous. Through simulation, the authors show that the error leads the sampler to prefer `sparse' models with fewer sinusoids. In this section, we first briefly outline both the model and the sampler proposed in \citet{andrieu_joint_1999}. After discussing the error noted by \citet{roodaki_comments_2013}, we employ the exact tests introduced in Section \ref{sec:exact} to show that this error could easily have been detected with our methods. This example demonstrates the utility of routine use of such tests in advance of publishing results that rely on estimation by MCMC.

\subsection{Model Description}

Consider a data vector $y := (y_0, \ldots, y_{N-1})^t$, which may aggregate multiple sinusoidal signals in addition to random noise. \citet{andrieu_joint_1999} propose  a series of competing models to explain such data, which are indexed by the number $k \geq 0$ of latent sinusoids hidden within the noisy signal. The $k$\textsuperscript{th} ($k > 0$) model is
\begin{equation*}
    y_{t} = \sum_{l=1}^k \left(c_{k,l}\cos(w_{k,l}t ) + s_{k,l} \sin(w_{k,l} t) \right) + \epsilon_{k, t},
\end{equation*}
where $\epsilon_{k,t}$ is white Gaussian noise with variance $\sigma^2_k$. The zero\textsuperscript{th} model corresponds to no latent signal, i.e. $y_t = \epsilon_{0,t}$. It is convenient, particularly when we discuss the priors, to express these models in matrix-vector form
\begin{equation*}
    y = \text{D}(w_k) a_k + \epsilon_k,
\end{equation*}
with radial frequencies $w_k := (w_{k,1}, \ldots, w_{k,k})^t$, amplitudes $a_k := (c_{k,1}, s_{k,1}, \ldots, c_{k,k}, s_{k,k})^t$, and noise $\epsilon_k := (\epsilon_{k,0}, \ldots, \epsilon_{k,N-1})^t$. The $N \times 2k$ design matrix $\text{D}(w_k)$ has odd-column entries $\text{D}(w_k)_{t+1,2l-1} = \cos(w_{k,l}t)$ and even-column entries $\text{D}(w_k)_{t+1,2l} = \sin(w_{k,l}t)$. 

The number of latent signals $k$ is given a Poisson prior with rate $\Lambda$, truncated to $\{0,\ldots, k_{max}\}$, where $k_{max} = \lfloor(N-1)/2\rfloor$. This upper limit prevents dependence in the columns of $D$, which would render $w_k$ difficult to identify. Conditional on $k$, the remaining priors are
\begin{align*}
    \sigma^2_k \mid k &\sim \text{InvGamma}(v_0/2, \gamma_0/2) \\
    w_k \mid k & \sim \text{Uniform}((0,\pi)^k) \\
    a_k \mid k, w_k, \sigma^2_k &\sim \mathcal{N}(0, \delta^2\Sigma_k),
\end{align*}
where $v_0$ and $\gamma_0$ are shape and scale parameters, and $\Sigma_k = \sigma^2_k\left(\text{D}^t(w_k)\text{D}(w_k)\right)^{-1}$. The prior on $a_k$ is known as the \textit{g-prior}.

\subsection{The Sampler and its Error}

\citet{andrieu_joint_1999} integrate $a_k$ and $\sigma^2_k$ out of the posterior and target $p(k, w_k | y)$, which is known up to a normalizing constant, and defined on $\Theta := \cup_{k=0}^{k_{max}} \{k\} \times (0,\pi)^k$. Their sampler employs four Markov kernels. Two are \textit{within-model} kernels, designed to update the frequencies $w_k$ while keeping $k$ fixed. The remaining two, namely the `birth' and `death' moves, are \textit{between-model} kernels, and propose moves that traverse subspaces of different dimensions; adding and deleting sinusoids respectively. 

The within-model kernels are standard Metropolis Hastings updates and alter only one component of $w_k$ at a time. The first perturbs the current state with a symmetric Gaussian proposal with scale $\sigma_{rw}$, and is referred to hereon-in as the local frequency kernel (LFK). The other uses a global proposal based on the Fourier coefficients of $y$. These are computed with the discrete Fourier transform, however, $y$ can be padded to length $N_p > N$ in order to improve interpolation. We refer to this as the global frequency kernel (GFK).

Birth and death kernels allow the chain to transition from dimension $k$ to $k+1$ or $k-1$ respectively. Suppose the current state of the chain is $(k, w_k)$. A birth move proposes a new frequency uniformly on $(0,\pi)$ and appends it to $w_k$, while a death move attempts to delete a component of $w_k$ randomly. To ensure reversibility with respect to $p(k, w_k | y)$, the authors propose to accept a birth proposal with probability $\min(1,r)$, where
\begin{equation} \label{eq:r_orig}
    r = \left( \frac{\gamma_0 + y^t\text{P}_ky}{\gamma_0 + y^t\text{P}_{k+1}y}\right)^{(N+v_0)/2} \frac{1}{(k+1)(1+\delta^2)},
\end{equation}
and
\begin{equation*}
    \text{P}_k = \text{I}_{N} - \frac{\delta^2}{1 + \delta^2} \text{D}(w_k)\left(\text{D}(w_k)^t \text{D}(w_k)\right)^{-1}\text{D}(w_k)^t,
\end{equation*}
for $k \geq 1$ and $\text{P}_0 = \text{I}_N$. Similarly, a death move is accepted with probability $\min(1,r^{-1})$.

\citet{roodaki_comments_2013} prove that \eqref{eq:r_orig} is erroneous and that the ratio $r$ should be replaced by $(1+k)r$. The authors demonstrate that the error leads the sampler to target the posterior where the prior on $k$ is 
\begin{equation}
    p(l) \propto \frac{e^{-\Lambda} \Lambda^l}{\left(l!\right)^2},
\end{equation}
for $l \in \{0,\ldots, k_{max}\}$. This is an \textit{accelerated Poisson} distribution, and places greater mass on small values. The implication is that all articles using the erroneous sampler place more emphasis on sparse models than intended.

\subsection{Testing the Sampler}

Our approach is to test the constituent kernels of the sampler individually, insofar as this is possible. We first test LFK and GFK separately on a problem with a known number of sinusoids, i.e. where $k$ is fixed. In the second stage, $k$ is treated as unknown, and the overall RJ-MCMC sampler is tested. It is difficult to test the birth and death kernels individually because they have a state-dependent selection probability, and are not irreducible by themselves. 

\subsubsection{Testing the Within-Model Kernels}
\label{sec:freq_test}

LFK and GFK are tested separately using the sequential two-sample and rank tests. Throughout, we assume that there is one sinusoid, i.e. $k=1$, and so the frequency $w_1$ reduces to a scalar in $(0,\pi)$. This could be extended to $k > 1$ by embedding the kernels in random scan samplers in order to update all frequency components. To fully define the joint distribution of data and parameters, we set $N = 64$, $v_0 = 10$, $\gamma_0 = 10$, and $\delta = 8$. Ideally, the kernel parameters $\sigma_{rw}$ and $N_p$ should be set so that the kernels mix fast and the tests have high power. Through experimentation, we determined that $\sigma_{rw} = 1/50$ led to LFK mixing well. For GFK, we use zero padding by letting $N_p = 4N$. This helps to better interpolate the Fourier frequencies, and empirically leads to lower rejection rates. 

We use $w_1$ as the sole test function in all tests. For the two-sample test we let $N_1=N_2=10^4$, and propagate the direct samples $L = 10^2$ steps with the kernels. The two-sample KS test is used to compare the direct and indirect samples. Given, however, that $w_1$ is uniform on $(0,\pi)$ under the prior, one could replace the direct samples with the uniform distribution function, and instead employ a one-sided KS test. For the rank tests, we use $10^4$ replications, $L=10$ and thinning of $10$. All tests use the default sequential parameters described in Section \ref{sec:sequential}; ensuring a false rejection rate of less than $10^{-5}$. 

The sequential two-sample test applied to LFK gave a p-value of $0.076$ in the first iteration, which triggered a second sequential iteration using $4 \times 10^4$ iterations. This resulted in a p-value of 0.56 and so no inconsistency was detected. Applying the same test to GFK led to a p-value of 0.36 in the first iteration. The sequential rank test gave p-values of 0.55 and 0.69 for the LFK and GFK kernels respectively in the first iteration. Therefore, in all cases, no error was detected.

\subsubsection{Testing the Full Sampler}
\label{sec:full_test}

Having tested the within-model kernels individually, we now test the full sampler that includes LFK, GFK, birth and death moves. For details on the overall algorithm please refer to \citet{andrieu_joint_1999}. All tests use the number of sinusoids $k$ as the test function. For the two-sample tests we let $N_1 = N_2 = 10^3$, and for the rank tests we use $10^3$ replications. All other parameter values remain the same as in the previous section.

We first test the original, erroneous sampler using the truncated Poisson prior on $k$ with rate $\Lambda = 3$. Both the two-sample and rank tests failed on the first iteration with p-values of $1.4 \times 10^{-6}$ and $2 \times 10^{-173}$ respectively, all but proving that the sampler contains an error.  The top panel of Figure \ref{fig:rjmcmc_results} shows the results of these tests. The empirical distribution of $k$ from the indirect samples is skewed to the left, as we expect given \citet{roodaki_comments_2013}'s finding that the sampler is in effect assuming an accelerated Poisson prior on $k$. There is also a strong skew in the rank statistics. We performed the same tests using the accelerated Poisson prior on $k$. As expected, no discrepancy was detected in this case. 

\begin{figure*}[tbp]
    \centering
    \resizebox{\textwidth}{!}{
    \input{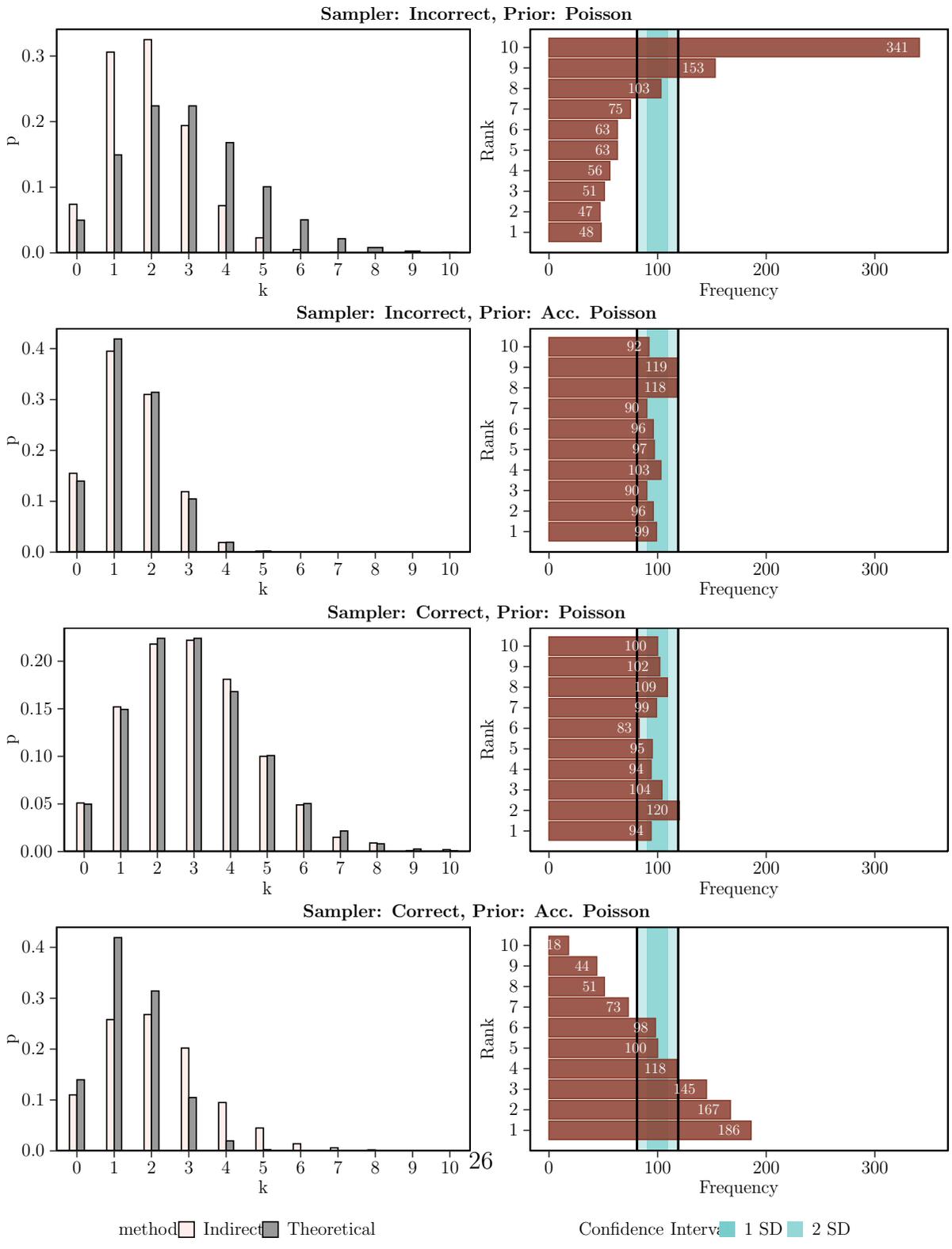}
    }
    \caption{\small Results of the two-sample and rank tests applied to the full RJ-MCMC sampler.}
    \label{fig:rjmcmc_results}
\end{figure*}

Next we replaced the ratio $r$ with $(1+k)r$ in the acceptance rate of birth and death moves, as suggested by \citet{roodaki_comments_2013}. Neither test detected an error when the correct truncated Poisson prior was used, however errors were detected when the accelerated Poisson prior was used. All results are shown in Figure \ref{fig:rjmcmc_results}.

\section{DISCUSSION}
\label{sec:conclusion}

This article has proposed two tests of MCMC implementations, which
 are unique in being exact; that is, the false rejection rate can be controlled. 
 This property is leveraged to propose a sequential testing procedure which 
 allows for high power and arbitrarily low false rejection rates, for example 
 $10^{-5}$. Such a procedure is useful for unit testing both MCMC and 
 more general Monte Carlo implementations, where one wants to minimize 
 the risk of rejecting a correct sampler. 
 
The performance of the two tests has been tested in a simulation study,
 and compared to other methods in the literature. The study validates the 
 ability of the tests to achieve the correct nominal level, and generally shows 
 favorable performance of the methods. The exact rank test is shown to have 
 high power over other methods when there are large errors within each conditional 
 distribution which may not aggregate to an easily detectable error in the joint 
 distribution of data and parameters. On the flip side, we have tested small 
 errors in the conditionals which are detectable in the joint. In this latter case, 
 the Geweke method appears to have a power advantage. However, we have 
 demonstrated extensions to the tests which improve their power.

\bibliographystyle{apalike}
\bibliography{references} 

\section*{APPENDIX A: PROOFS}

\renewcommand{\thesubsection}{\Alph{subsection}}
\label{sec:proofs}

\begin{proof}[Proof of Lemma \ref{lemma:besag}]
For $k\in 1,\dots,N$, 
\begin{equation*}
\begin{split}
\prob \{R_M = k \} & = \sum_{m=1}^{N} \prob \{R_M = k \mid  M = m\} \prob \{M=m\}\\
& = \frac{1}{N} \sum_{m=1}^{N}  \prob \{R_m = k \} = \frac{1}{N},
\end{split}
\end{equation*}
where the last equality holds because exactly one element of $(R_1,...,R_N)$ must equal $k$. In other words, the events $\{R_m = k\}$ partition the sample space.
\end{proof}

\begin{proof}[Proof of Proposition \ref{prop:uniform_null}]
Consider the variables involved in one evaluation of Algorithm \ref{alg:ranks}. The proposition assumes that the kernel $K_y$ is $\pi(\theta \mid  y)$-reversible. Letting $f$ denote the joint distribution of $\theta_{1:N}$ then for $m > 1$,
\begin{equation*}
\begin{split}
f(\theta_{1:N} \mid M = m, y) & = \pi(\theta_m \mid y) \left( \prod_{i=1}^{m-1} K_y(\theta_{m-i+1}, \theta_{m-i} ) \right) \left( \prod_{j=m+1}^{N} K_y(\theta_{j-1}, \theta_{j} ) \right)\\
& = \pi(\theta_m \mid y) K(\theta_m, \theta_{m-1}) \left( \prod_{i=1}^{m-2} K_y(\theta_{m-1 -i+1}, \theta_{m-1 + i} ) \right) \left( \prod_{j=m+1}^{N} K_y(\theta_{j-1}, \theta_{j} ) \right)\\
& = \pi(\theta_{m-1} \mid y) K(\theta_{m-1}, \theta_{m}) \left( \prod_{i=1}^{m-2} K_y(\theta_{m-1 -i+1}, \theta_{m-1 + i} ) \right) \left( \prod_{j=m+1}^{N} K_y(\theta_{j-1}, \theta_{j} ) \right)\\
& = \pi(\theta_{m-1} \mid y) \left( \prod_{i=1}^{m-2} K_y(\theta_{m-1 -i+1}, \theta_{m-1 + i} ) \right) \left( \prod_{j=m}^{N} K_y(\theta_{j-1}, \theta_{j} ) \right)\\
&= f(\theta_{1:N} \mid M = m-1, y).
\end{split}
\end{equation*}
This implies that the distribution of $\theta_{1:N}$ is independent of $M$. Since additionally $R$ is also assumed independent of $M$,  both $R(\theta_{1:L})$ and $M$ satisfy the conditions of Lemma \ref{lemma:besag} and so $R_M(\theta_{1:L})$ is uniformly distributed.
\end{proof}

\begin{proof}[Proof of Theorem \ref{thm:size_bound}] We use induction from $j=k$ to $j=1$ to show that
  \begin{equation}
    \label{eq:indhyp}
   \Prob \{ \text{fail} \mid \text{step } j \text{ reached} \}\leq
  (k+1-j)\beta_{j} ,
  \end{equation}
 where $\beta_j=\beta/\gamma^{j-1}$ is as defined in Algorithm \ref{alg:sequential}.

 First, by the usual arguments for the Bonferroni correction, 
    $\Prob \{ q_{i}\leq p \} \leq p$ for all $p \in [0,1]$ and for $i=1,\dots,k$.
This,  immediately shows that (\ref{eq:indhyp}) holds for  $j=k$.

    To show that (\ref{eq:indhyp})  holds for $j=i\in \{1,\dots,k-1\}$ given that it holds for $j=i+1$, we argue as follows.
    Let $A_{i}=\{\beta_{i}<q_i\leq \gamma+\beta_i\}$ and $B_i=\{q_i\leq \beta_{i}\}$. Using the arguments for the Bonferroni correction again gives $\Prob \{B_i \}\leq \beta_i$
    and $\Prob \{ A_i \}\leq \gamma+\beta_i-\Prob \{B_i \}$.
Then 
\begin{align*}
  \Prob \{&\text{fail} \mid \text{step } i \text{ reached}\}=\Prob\{B_{i}\}+\Prob\{A_{i}, \text{fail} \mid \text{step }i\text{ reached}\}\\
  &=\Prob \{B_{i} \}+\Prob\{A_{i}\}\Prob\{\text{fail} \mid \text{step i+1 reached}\}\\
  &\leq\Prob\{B_{i}\}+(\gamma+\beta_i-\Prob\{B_i)\}\Prob\{\text{fail} \mid \text{step i+1 reached}\}\\
  &=\gamma \Prob\{\text{fail} \mid \text{step i+1 reached}\} + \beta_i\Prob\{\text{fail} \mid \text{step i+1 reached}\}+\Prob\{B_{i}\}(1-\Prob\{\text{fail} \mid \text{step i+1 reached}\})\\
        &\leq \gamma \Prob \{\text{fail} \mid \text{step } i+1 \text{ reached} \} + \beta_i\\
          &\leq\gamma (k+1-(i+1))\beta_{i+1}+\beta_{i}=(k+1-i)\beta_i
\end{align*}
Thus using  (\ref{eq:indhyp}) for $i=1$ gives $\Prob \{ \text{fail} \} \leq k\beta_1=k\beta=\alpha.$    
  \end{proof}

\section*{APPENDIX B: TUNING SEQUENTIAL PARAMETERS}  
\label{sec:tuning}
 
We use a simulation study to propose default parameters for the sequential 
 tests. The classical goodness-of-fit setting is considered: independent and identically distributed 
 samples from can be generated and the task is to test if the samples derive from a
standard normal distribution.  The two-sided
Kolmogorov-Smirnov test is used to test this.

The sample size for $k=1$ and $\Delta=1$ (i.e. the non-sequential
setting) was chosen to be $10^4$. Sample sizes for other settings
were adjusted using (\ref{eq:effort_bound}) so that under the null the
computational effort were identical. $\alpha$ is set at $10^{-5}$ to replicate the situation where
we only want very few rejections. 

\begin{table}[tbp]
  \caption{\small Power of the sequential procedure using a KS test on iid data. Null is $\mathcal{N}(0,1)$.}\label{tab:seqrejprob}
  \footnotesize
  \centering
\begin{tabular}{rrrrrrrr}
  \hline 
 & N(0,1),$\alpha$=0.01 & N(0,1) & N(0.05,1) & N(0.03,1) & N(0.02,1) & N(0,$0.95^2$) & N(0,$0.97^2$) \\ 
  \hline
k=1,$\Delta$=1 & 0.011 & 0.000 & 0.415 & 0.028 & 0.003 & 0.007 & 0.000 \\ 
  k=3,$\Delta$=1 & 0.010 & 0.000 & 0.939 & 0.202 & 0.016 & 0.380 & 0.004 \\ 
  k=3,$\Delta$=2 & 0.009 & 0.000 & 0.967 & 0.448 & 0.061 & 0.681 & 0.025 \\ 
  k=3,$\Delta$=4 & 0.009 & 0.000 & 0.960 & 0.529 & 0.156 & 0.658 & 0.109 \\ 
  k=5,$\Delta$=1 & 0.009 & 0.000 & 0.974 & 0.285 & 0.026 & 0.615 & 0.008 \\ 
  k=5,$\Delta$=2 & 0.010 & 0.000 & 0.986 & 0.583 & 0.100 & 0.861 & 0.080 \\ 
  k=5,$\Delta$=4 & 0.009 & 0.000 & 0.979 & 0.632 & 0.233 & 0.808 & 0.227 \\ 
  k=7,$\Delta$=1 & 0.009 & 0.000 & 0.988 & 0.392 & 0.035 & 0.876 & 0.035 \\ 
  k=7,$\Delta$=2 & 0.010 & 0.000 & 0.990 & 0.692 & 0.141 & 0.959 & 0.231 \\ 
  k=7,$\Delta$=4 & 0.011 & 0.000 & 0.975 & 0.702 & 0.286 & 0.887 & 0.408 \\ 
  k=9,$\Delta$=1 & 0.008 & 0.000 & 0.987 & 0.384 & 0.035 & 0.916 & 0.054 \\ 
  k=9,$\Delta$=2 & 0.009 & 0.000 & 0.990 & 0.673 & 0.124 & 0.963 & 0.266 \\ 
  k=9,$\Delta$=4 & 0.010 & 0.000 & 0.965 & 0.693 & 0.252 & 0.892 & 0.430 \\ 
  k=11,$\Delta$=1 & 0.010 & 0.000 & 0.985 & 0.364 & 0.027 & 0.919 & 0.063 \\ 
  k=11,$\Delta$=2 & 0.010 & 0.000 & 0.988 & 0.636 & 0.105 & 0.966 & 0.267 \\ 
  k=11,$\Delta$=4 & 0.008 & 0.000 & 0.949 & 0.674 & 0.198 & 0.891 & 0.420 \\ 
   \hline
\end{tabular}
\\
 $n=10^4$ for the non-sequential test ($k=1,\Delta=1$); other $n$ adjusted to give same expected effort under null,   $\alpha=10^{-5}$, unless otherwise indicated.
\end{table}

\begin{table}[tbp]
  \caption{\small Power of the sequential procedure using a KS test on iid data. Null is $\mathcal{N}(0,1)$.}\label{tab:seqrejprob1000}
  \footnotesize
  \centering
\begin{tabular}{rrrrrrrr}
  \hline
 & N(0,1),$\alpha$=0.01 & N(0,1) & N(0.15,1) & N(0.1,1) & N(0.05,1) & N(0,$0.85^2$) & N(0,$0.9^2$) \\ 
  \hline
k=1,$\Delta$=1 & 0.009 & 0.000 & 0.324 & 0.039 & 0.001 & 0.006 & 0.000 \\ 
  k=3,$\Delta$=1 & 0.009 & 0.000 & 0.902 & 0.249 & 0.004 & 0.349 & 0.006 \\ 
  k=3,$\Delta$=2 & 0.009 & 0.000 & 0.937 & 0.522 & 0.014 & 0.668 & 0.058 \\ 
  k=3,$\Delta$=4 & 0.009 & 0.000 & 0.939 & 0.576 & 0.050 & 0.658 & 0.162 \\ 
  k=5,$\Delta$=1 & 0.008 & 0.000 & 0.948 & 0.348 & 0.006 & 0.604 & 0.019 \\ 
  k=5,$\Delta$=2 & 0.010 & 0.000 & 0.971 & 0.655 & 0.022 & 0.848 & 0.145 \\ 
  k=5,$\Delta$=4 & 0.009 & 0.000 & 0.959 & 0.677 & 0.085 & 0.817 & 0.295 \\ 
  k=7,$\Delta$=1 & 0.011 & 0.000 & 0.973 & 0.481 & 0.006 & 0.887 & 0.077 \\ 
  k=7,$\Delta$=2 & 0.011 & 0.000 & 0.983 & 0.759 & 0.031 & 0.952 & 0.372 \\ 
  k=7,$\Delta$=4 & 0.009 & 0.000 & 0.958 & 0.744 & 0.095 & 0.890 & 0.487 \\ 
  k=9,$\Delta$=1 & 0.009 & 0.000 & 0.972 & 0.468 & 0.005 & 0.917 & 0.112 \\ 
  k=9,$\Delta$=2 & 0.009 & 0.000 & 0.985 & 0.738 & 0.025 & 0.962 & 0.412 \\ 
  k=9,$\Delta$=4 & 0.009 & 0.000 & 0.949 & 0.725 & 0.072 & 0.883 & 0.531 \\ 
  k=11,$\Delta$=1 & 0.008 & 0.000 & 0.968 & 0.441 & 0.004 & 0.921 & 0.127 \\ 
  k=11,$\Delta$=2 & 0.009 & 0.000 & 0.977 & 0.712 & 0.019 & 0.967 & 0.418 \\ 
  k=11,$\Delta$=4 & 0.009 & 0.000 & 0.936 & 0.722 & 0.047 & 0.883 & 0.525 \\ 
   \hline
\end{tabular}
\\
 $n=10^3$ for the non-sequential test ($k=1,\Delta=1$),   $\alpha=10^{-5}$, unless otherwise indicated.
\end{table}

Results are in Table \ref{tab:seqrejprob}, based on $10^4$ repeated
tests. The first two columns are under the null i.e.\ we would only
expect the nominal number of rejections. This seems to be roughly the
case.

Table \ref{tab:seqrejprob1000} is similar to  Table \ref{tab:seqrejprob}, with the exception that the sample size for  $k=1$ and $\Delta=1$ (i.e. the non-sequential
setting) was chosen to be $10^3$ and that some different alternative have been considered. 

For the alternatives there is a very substantial increase in terms of
power compared to the non-sequential approach ($k=1,\Delta=1$).
Increasing the sample size at the second step seems beneficial -
$\Delta=2$ and $\Delta=4$ seem to be doing better than $\Delta=1$ in
the simulation results.  Furthermore, the number of sequential steps should be large (at least $k\geq 5$).

An overall good performance seems to be achieved by using $k=7$ and
$\Delta=4$.  Therefore, these are the default settings used in our
R-package.

\end{document}